\documentclass{article}

\usepackage{microtype}
\usepackage{graphicx}
\usepackage{subcaption}
\usepackage{booktabs} 

\usepackage{amsmath,amsfonts,bm}

\def\eqref#1{equation~\ref{#1}}

\def\1{\bm{1}}

\DeclareMathAlphabet{\mathsfit}{\encodingdefault}{\sfdefault}{m}{sl}
\SetMathAlphabet{\mathsfit}{bold}{\encodingdefault}{\sfdefault}{bx}{n}

\newcommand{\E}{\mathbb{E}}

\newcommand{\R}{\mathbb{R}}

\newcommand{\Var}{\mathrm{Var}}

\newcommand{\Cov}{\mathrm{Cov}}

\DeclareMathOperator*{\argmin}{arg\,min}

\usepackage{hyperref}
\usepackage{url}
\usepackage{amsmath, amssymb, amsthm}
\usepackage{array}
\usepackage{mathtools}
\usepackage{booktabs}
\usepackage{makecell}
\usepackage{graphicx}
\usepackage{thmtools}
\usepackage{float}
\usepackage{subcaption} 
\usepackage{enumitem}
\usepackage{thm-restate}
\usepackage{tikz}
\usetikzlibrary{decorations.pathreplacing}

\usetikzlibrary{shadows.blur}
\usetikzlibrary{calc}
\usetikzlibrary{arrows.meta}
\usepackage{hyperref}

\newcommand{\RelVar}{\mathrm{RelVar}}

\usepackage[preprint]{icml2026}

\usepackage{amsmath}
\usepackage{amssymb}
\usepackage{mathtools}
\usepackage{amsthm}

\usepackage[capitalize,noabbrev]{cleveref}

\theoremstyle{plain}
\newtheorem{theorem}{Theorem}[section]

\newtheorem{lemma}[theorem]{Lemma}
\newtheorem{corollary}[theorem]{Corollary}
\theoremstyle{definition}
\newtheorem{definition}[theorem]{Definition}

\theoremstyle{remark}
\newtheorem{remark}[theorem]{Remark}

\newtheoremstyle{restate}   
  {}{}                      
  {\itshape}                
  {}                        
  {\bfseries}               
  {}                        
  { }                       
  {}                        

\theoremstyle{restate}
\newtheorem*{restate}{}
\usepackage[textsize=tiny]{todonotes}

\icmltitlerunning{Breaking the Curse of Dimensionality: On the Stability of Modern Vector Retrieval}

\begin{document}

\twocolumn[
  \icmltitle{Breaking the Curse of Dimensionality: \\ On the Stability of Modern Vector Retrieval}



  \icmlsetsymbol{equal}{*}

  \begin{icmlauthorlist}
    \icmlauthor{Vihan Lakshman}{yyy}
    \icmlauthor{Blaise Munyampirwa}{comp}
    \icmlauthor{Julian Shun}{yyy}
    \icmlauthor{Benjamin Coleman}{sch}
  \end{icmlauthorlist}

  \icmlaffiliation{yyy}{MIT CSAIL, USA}
  \icmlaffiliation{comp}{Argmax, Inc., USA}
  \icmlaffiliation{sch}{Google DeepMind, USA}

  \icmlcorrespondingauthor{Vihan Lakshman}{vihan@mit.edu}

  \icmlkeywords{Machine Learning, ICML}

  \vskip 0.3in
]



\printAffiliationsAndNotice{}  

\begin{abstract}
Modern vector databases enable efficient retrieval over high-dimensional neural embeddings, powering applications from web search to retrieval-augmented generation. However, classical theory predicts such tasks should suffer from \emph{the curse of dimensionality}, where distances between points become nearly indistinguishable, thereby crippling efficient nearest-neighbor search. We revisit this paradox through the lens of \emph{stability}, the property that small perturbations to a query do not radically alter its nearest neighbors. Building on foundational results, we extend stability theory to three key retrieval settings widely used in practice: (i) multi-vector search, where we prove that the popular Chamfer distance metric preserves single-vector stability, while average pooling aggregation may destroy it; (ii) filtered vector search, where we show that sufficiently large penalties for mismatched filters can induce stability even when the underlying search is unstable; and (iii) sparse vector search, where we formalize and prove novel sufficient stability conditions. Across synthetic and real datasets, our experimental results match our theoretical predictions, offering concrete guidance for model and system design to avoid the curse of dimensionality.

\end{abstract}

\section{Introduction}

The \emph{curse of dimensionality} is a notorious barrier to algorithmic efficiency in computational geometry. Informally, the curse states that points in high-dimensional metric spaces tend to be arbitrarily close in distance. Consequently, this property makes it challenging for techniques such as near-neighbor search and near-neighbor classification to identify the closest points to a query without resorting to a full brute-force scan, which can be prohibitively slow in practice. Nevertheless, modern vector databases, powered by stunning advances in neural representation learning, have achieved remarkable adoption in spite of high vector dimensionalities, powering applications such as retrieval-augmented generation \citep{lewis2020retrieval}, web-search \citep{huang2020embedding}, and product search \citep{nigam2019semantic}. 

The success of modern vector retrieval raises a question that we find surprisingly unexplored in the current research literature: \emph{Given that vector databases efficiently search high-dimensional embedding spaces, what properties do these search algorithms and embeddings possess that make them amenable to sub-linear time search algorithms?} In addition to enhancing our theoretical understanding of the curse of dimensionality, successfully answering this question also has implications for machine learning model developers and vector database engineers. In particular, if we can identify necessary or sufficient conditions for embedding-based search to ``break'' the curse of dimensionality, practitioners can potentially leverage these insights when working to optimize models and search algorithms by maintaining the essential invariants that account for the success of these representations in high dimensions. 

In this paper, we study this question from a theoretical and experimental lens, using the formal notion of \emph{stability} to characterize the conditions under which the curse of dimensionality is overcome. We focus on three vector retrieval settings widely used in practice: 
\begin{enumerate}
\item \textbf{Multi-Vector Search} which involves performing near-neighbor search in the case where both the query and the database items consist of \emph{sets} of vectors. This style of search has become extremely popular due to the recent success of ``late interaction" models such as ColBERT \citep{khattab2020colbert}. 

\item \textbf{Filtered Vector Search} which involves near neighbor search in the presence of attribute filters. This setting has also become widely deployed in vector databases in practice \citep{gollapudi2023filtered, wang2023efficient} due to the need to filter results based on external metadata. 

\item \textbf{Sparse Vector Search} which is near neighbor search in the special case where the embedding vectors are extremely high-dimensional (such as in the tens or hundreds of thousands) but only a small number of components are non-zero. This search style has become popular due to the success of learned sparse representation models such as SPLADE \citep{formal2021spladev2, formal2021splade} and semantic models in recommender systems \citep{rajput2023recommender, penha2025semantic}.

\end{enumerate}

Intuitively, stability captures the notion that a slight perturbation to a query should not dramatically change the list of near neighbors to that query. The landmark paper of \citep{beyer1999nearest} and a subsequent follow-up work of \citep{durrant2009nearest} together establish a necessary and sufficient condition for stability via a quantity called \emph{relative variance}. However, to our knowledge, no work has studied extending these prior results to modern vector database settings, particularly in the three cases that we focus on in this paper. Thus, our work fills a gap in the literature in understanding why modern neural representations enable efficient sub-linear vector search in high dimensions. Our specific contributions are as follows:

\begin{itemize}

\item We prove that multi-vector search with the commonly-used Chamfer distance is stable if the \emph{induced} single vector search problem is \emph{strongly stable} (along with some technical conditions on the composition of the vector sets). However, not all aggregations are stability preserving; we give a counter-example showing that stability does not hold in general for average pooling. This result provides a theoretical justification for the choice of the Chamfer distance used in the popular ColBERT family of models \citep{khattab2020colbert, santhanam2022plaid, faysse2025colpali}. 

\item We prove that filtered vector search in the \emph{distance penalty} model is stable provided that the penalty function applied to points that do not satisfy the filter is sufficiently large. Interestingly, we show that, for a suitably large penalty, filtered search can be stable even if the underlying vector search problem (without filters) is unstable. We also discuss why this distance penalty model encompasses the majority of hybrid vector search algorithms in the literature. 

\item In the sparse vector search setting, we present the first formal definition of a property called \emph{concentration of importance}, which was intuitively described in a paper of \citep{bruch2024seismic}, as well as the definition of a new property we call \emph{overlap of importance}. We then prove that concentration of importance and overlap of importance (along with some minor technical conditions preventing non-degeneracy) together imply stability.

\item In addition to our theoretical results, we also empirically validate each of these theorems on a combination of synthetic and real embedding datasets\footnote{https://github.com/vihan-lakshman/ann-stability-theory}. 

\end{itemize}

\subsection{Additional Results in Appendix}

We defer several results to the appendix. In Appendix \ref{sec:intrinsic-dim}, we discuss prior work on \emph{intrinsic dimensionality}, a related concept to the theory of stability. In Appendix \ref{sec:practical-implications}, we demonstrate empirically that stability is not merely a theoretical curiosity but rather has substantial bearing on the quality of vector databases in practice. In the remaining Appendices (\ref{sec:multi-vec-appendix}, \ref{sec:filtered-appendix}, and \ref{sec:sparse-stability-appendix}), we provide full proofs for all of the new mathematical results introduced in this paper, and we empirically verify that the conditions of our theorems hold on real embedding datasets. 

\section{Background}

In this section, we introduce the necessary mathematical preliminaries for our main results. We begin with a formal definition of nearest-neighbor search.

\begin{definition}[Near-Neighbor Search Problem]

A \emph{near-neighbor search problem} is a 3-tuple $(Q, D, \delta)$ for a query set $Q$ and database collection $D$ where $q, d \in \mathbb{R}^m$ for all $q \in Q$ and $d \in D$. For each $q \in Q$, we use the distance metric $\delta\colon \mathbb{R}^m \times \mathbb{R}^m \to \mathbb{R}_{\ge 0}$ to compute $$\argmin_{d \in D} \delta(q, d)$$

\end{definition}

In the remainder of this paper, we will assume that all database elements $d \in D$ and all query elements $q \in Q$ are sampled from two respective distributions. Under this probabilistic model, the distance $\delta(q, d)$ for $q \sim Q$ and $d \sim D$ is a random variable.  In a seminal paper, \citet{beyer1999nearest} introduced this probabilistic model to provide a precise formalization of the curse of dimensionality in the context of near-neighbor search by defining the notion of instability.

\begin{definition}[Instability \citep{beyer1999nearest}]

A nearest neighbor problem $(Q, D, \delta)$ is \emph{unstable} if for all $\epsilon > 0$
$$\lim_{m \to \infty} 
\mathop{\Pr}\limits_{q \sim Q}
\left[
  \max_{d \in D} \{\delta(q, d)\}
  \le (1 + \epsilon)
  \min_{d \in D} \{\delta(q, d)\}
\right]
= 1$$

If this property does not hold, we say the problem is \emph{stable}. 
\end{definition}

\begin{remark}
For notational convenience, we will not explicitly index the random variable $\delta(q, d)$ by the dimensionality, but it is important to remember that $\delta(q, d)$ is a function of the dimension $m$. 
\end{remark}

We can also equivalently characterize instability as the convergence of the ratio of the maximum instance to the minimum distance as we take the dimensionality to infinity. 

\begin{corollary}[\citet{beyer1999nearest}]
For a nearest neighbor problem $(Q, D, \delta)$, let $q \sim Q$ and define $\mathrm{DMAX} = \max_d \delta(q, d)$ and $\mathrm{DMIN} = \min_d \delta(q, d)$. A nearest neighbor problem is unstable if and only if $\frac{\mathrm{DMAX}}{\mathrm{DMIN}} \to_p 1$ where the sequence is taken over the dimensionality.  
\end{corollary}

In other words, a nearest neighbor problem is unstable if the distance between the closest and farthest points across all queries becomes nearly indistinguishable. We can also interpret stability as the property that small perturbations to a query do not dramatically alter its list of near neighbors. This is a mathematical formalization of the curse of dimensionality, which states that high-dimensional spaces lack the contrast needed for meaningful near-neighbor search. In this setting, sublinear approximate search algorithms fail, and a brute force linear scan becomes necessary. 



To prove that a near-neighbor problem is stable or unstable, we will also stand on the shoulders of two prior theorems from \citep{beyer1999nearest} and \citep{durrant2009nearest}, respectively. 

\begin{theorem}[\citet{beyer1999nearest}]

If ~ $\lim_{m \to \infty} \frac{\Var\left[ \delta(q, d)\right]}{\E\left[\delta(q, d) \right]^2} = 0$, then the nearest neighbor problem is unstable. 
\label{beyerthm}
\end{theorem}

\begin{theorem}[\citet{durrant2009nearest}]

For a sufficiently large database size $|D|$, if ~ $\lim_{m \to \infty} \frac{\Var\left[ \delta(q, d)\right]}{\E\left[\delta(q, d) \right]^2} \ne 0$, then the nearest neighbor problem is stable.
\label{durrantconverse}
\end{theorem}

Following the convention established in \citep{beyer1999nearest}, we refer to the quantity $\frac{\Var\left[ \delta(q, d)\right]}{\E\left[\delta(q, d) \right]^2}$ as the \emph{relative variance} or RelVar. In this paper, we will utilize both the convergence in probability characterization and the relative variance criterion to prove stability. 


\section{Related Work}

\subsection{Conditions for Stability}

The paper of \citet{beyer1999nearest} was a foundational work that inspired significant follow-up research in understanding the conditions for stable vector search. \citet{aggarwal2001surprising} study stability under common $\ell_p$ norm distance metrics, where they find that $\ell_1$ and fractional distance metrics are theoretically and practically favorable to $\ell_2$ distances. Furthermore, \citet{houle2010can} prove that the so-called \emph{shared-neighbor} distance function, which measures the number of near-neighbors in common between two points, is a more stable measure than absolute distance since the latter loses discriminatory power in high dimensions. In a different vein, \citet{hui2021nearest} study stability in the context of sequential data, such as genomes and sensor readings, and show that, under uniformly distributed sequences, the RelVar under edit distance goes to zero in the limit. Moreover, \citet{pestov2013k} extends the results of \citet{beyer1999nearest} from $k$-nearest neighbor \emph{search} to $k$-nearest neighbor \emph{classification} where they find that instability of vector search also implies barriers to the learnability of near-neighbor classification. 

\subsection{Multi-Vector Search}

The first retrieval setting we consider in this work is multi-vector search, where each query and document may be represented as \emph{sets} of vectors under a chosen set-similarity function. This style of vector retrieval first gained traction with the celebrated ColBERT paper \citep{khattab2020colbert}. Since the introduction of ColBERT, multiple works have studied how to accelerate of multi-vector search algorithms, including PLAID \citep{santhanam2022plaid}, DESSERT \citep{engels2023dessert}, and MUVERA \citep{jayaram2024muvera}. In addition, \cite{gollapudialpha} demonstrate that multi-vector search under the commonly-used Chamfer distance \citep{barrow1977parametric} can be integrated into the existing graph-based DiskANN index \citep{jayaram2019diskann} originally conceived for single-vector search. Furthermore, \citet{faysse2025colpali} extend the multi-vector modeling paradigm pioneered by ColBERT beyond text to multimodal vision-language retrieval. However, despite the tremendous research interest in multi-vector search, we note that no prior work has studied how the use of vector sets relates to the curse of dimensionality. We address this gap in the literature. 

\subsection{Filtered Similarity Search}

With the advent of vector databases in retrieval applications, there has also been a surge of interest in augmenting traditional near-neighbor seach indexes with the capability of filtering points according to additional metadata. For instance, an e-commerce search service might want to return relevant near-neighbor results subject to additional attributes associated to each point such as ``free shipping eligible." Consequently, a number of works have proposed techniques for augmenting traditional vector search algorithms with filtering capabilities, such as NHQ \citep{wang2023efficient}, Airship \citep{zhao2022constrained}, Filtered DiskANN \citep{gollapudi2023filtered}, Acorn \citep{patel2024acorn}, CAPS \citep{gupta2023caps}, IVF$^2$ \cite{landrum2025ivf}, and approximate window search~\citep{engels2024}. 

A unifying thread present in many of these approaches is the introduction of a \emph{penalty function} to augment the standard vector distance metric in the case where a given point does not match the query filter predicate. We adopt this penalty function model of filtered similarity search in our theoretical stability analysis. Ultimately, we find that suitably large penalty values provides a sufficient condition for vector search stability.

\subsection{Sparse Embedding Search}

We conclude this section by turning our attention to a third form of embedding-based retrieval that has recently gained significant interest, namely sparse vector search. A sparse vector is one which is typically extremely high-dimensional (at least tens of thousands of components) but with only a small number of non-zero entries. In fact, sparse vectors predate the recent advent of neural representation learning as traditional keyword retrieval via inverted indexes \citep{schutze2008introduction} can also be viewed as sparse vector search. Recently, sparse vectors have seen a renewal of interest in the form of \emph{learned sparse embeddings} as pioneered by the influential SPLADE \citep{formal2021splade, formal2021spladev2} family of models that combine the generalization power of neural representations with the interpretability properties of sparse vectors. Introduced in 2024, the Seismic inverted index-based approach \citep{bruch2024seismic} remains the current state-of-the-art for efficient and scalable sparse vector search. The key insight driving the design of Seismic is a novel observation that learned sparse embeddings, such as SPLADE, exhibit a property called concentration of importance, which states that the total $\ell_1$ mass of the sparse vector is primarily determined by a small subset of the non-zero coordinates. Inspired by this prior work, we formalize the notion of concentration of importance in this paper and draw a connection between this concept and stability. Although it turns out to not be sufficient for sparse vector search stability, we show that, under the additional assumption of \emph{overlap of importance} and some mild technical assumptions, we can prove stability.

\section{Stability of Multi-vector Search}

We now turn our attention to the stability of multi-vector search. We begin by introducing a formal mathematical model of multi-vector search including a new, analogous definition of stability in this setting. We then proceed to our main theorem of this section identifying a sufficient condition for such stability to exist. 




\subsection{Theoretical Model}

\begin{definition}[Multi-Vector Search Problem]
A multi-vector search problem is specified by a 4-tuple $(\mathcal{Q}, \mathcal{D}, \delta, \mathrm{Agg})$, where:
\begin{itemize}
    \item $\mathcal{Q}$ is a collection of \emph{query sets}. Each query $Q \in \mathcal{Q}$ is a finite set of vectors.
    \item $\mathcal{D}$ is a database of \emph{document sets}. Each document $D \in \mathcal{D}$ is a finite set of vectors.
    \item $\delta: \mathbb{R}^m \times \mathbb{R}^m \rightarrow \mathbb{R}_{\ge 0}$ is a \emph{primitive distance metric} between individual vectors.
    \item $\mathrm{Agg}$ is a \emph{set aggregation function} that computes the distance between a query set $Q$ and a document set $D$, based on the primitive distances between their constituent vectors.
\end{itemize}
For a given query $Q \in \mathcal{Q}$, the goal is to find:
$$
\arg\min_{D \in \mathcal{D}} \mathrm{Agg}(Q, D)
$$
\end{definition}

Below, we define two concrete aggregation functions that we consider in this work. 

\begin{definition}
For two finite vector sets $A$ and $B$ where $A, B \subset \mathbb{R}^m$, and a vector distance function $\delta$, the \emph{Chamfer Distance} \citep{barrow1977parametric} is defined as $$\mathrm{Chamfer}(A, B) =  \sum_{a \in A} \min_{b \in B} \delta(a, b)$$
\end{definition}

\begin{definition}
For two finite vector sets $A$ and $B$ where $A, B \subset \mathbb{R}^m$, and a vector distance function $\delta$, \emph{Average Pooling} aggregation is defined as $$\mathrm{AvgPool}(A, B) = \frac{1}{|A| \cdot |B|} \sum_{a \in A} \sum_{b \in B} \delta(a, b)$$
\end{definition}

Now, we will define a new notion of stability in the multi-vector search setting. 

\begin{definition}[Multi-Vector Search Stability]
\label{def:multi-vec-stability}
A multi-vector search problem is unstable if for all $\epsilon > 0$

$$\lim_{m \to \infty}\mathop{\Pr}\limits_{Q \sim \mathcal{Q}}\left[ \max_{D \in \mathcal{D}} \{\mathrm{Agg}(Q, D\} \le (1 + \epsilon) \min_{D \in \mathcal{D}} \{\mathrm{Agg}(Q, D)\} \right] = 1$$

\end{definition}

We can now state our first theorem of this section generalizing the results of \citep{beyer1999nearest} and \citep{durrant2009nearest} to the multi-vector setting.

\begin{theorem}
\label{thm:multi-vec-equivalence}
For sufficiently large databases, a multi-vector search instance is stable if and only if $$\lim_{m \to \infty} \frac{\Var\left[\mathrm{Agg}(Q, D)\right]}{\E\left[\mathrm{Agg}(Q, D) \right]^2} \ne 0$$
\end{theorem}

We defer the proof of this theorem to Appendix \ref{sec:multi-vec-appendix}. With this extension of stability to the multi-vector setting in hand, we will now define a procedure for constructing a single-vector search instance from a multi-vector problem.

\begin{definition}[Induced Single-Vector Search Problem]
Given a multi-vector search instance $(\mathcal{Q}, \mathcal{D}, \delta, \mathrm{Agg})$, we define the \emph{induced} single-vector search problem $(Q^\prime, D^\prime, \delta)$ as follows

\begin{itemize}
\item $Q^\prime = \bigcup_{Q_i \in \mathcal{Q}} Q_i$ is the set of unique query vectors, across all multi-vector query sets.

\item $D^\prime = \bigcup_{D_i \in \mathcal{D}} D_i$ is the set of unique document components, across all documents.

\item The distance metric $\delta$ is the primitive distance function of $\mathrm{Agg}$.
\end{itemize}
\end{definition}

\subsection{Theoretical Results}

From the previous definitions, we can now establish conditions for the stability of multi-vector search under Chamfer distance via the stability of the induced single vector search instance. The key idea behind our argument is to leverage the structure of the Chamfer distance, which, in essence, solves multiple scaled-down near neighbor search problems.

Before stating our main theorem of this section, we will need to define one additional property we call \emph{strong stability}. 

\begin{definition}[Strong Stability]
A near-neighbor search problem $(Q, D, \delta)$ is $c$-strongly stable for some constant $c > 1$ if $\frac{\mathrm{DMAX}}{\mathrm{DMIN}} > c + \Omega(1)$ almost surely.  
\end{definition}

Intuitively, strong stability asserts that a near neighbor problem is not only stable but also stable with a non-decreasing gap. We will need this stronger property when analyzing the stability of multi-vector search because, even if the induced single-vector search problem is stable, the organization of document vectors into sets might contract the maximum distance gap seen in the single vector search problem. Strong stability ensures that, as long as the vector sets obey a certain mild non-degeneracy condition, we retain stability even if the stability gap contracts.   

We will now proceed to our main results on the stability of multi-vector search. We will begin by proving a lemma establishing stability in the case where the query sets are singletons (each containing only one vector) while the document sets may be larger in size. Then, we will use this result to prove stability in the general case of both query and document sets with multiple elements. 

\begin{lemma}[Chamfer Stability with Singleton Query Sets]
\label{lem:limit_transfer}
Let \((\mathcal{Q}, \mathcal{D}, \delta, \mathrm{Chamfer})\) be a multi-vector search problem where $|Q_i| = 1$ for each query set $Q_i \in \mathcal{Q}$. Let \((Q', D', \delta)\) denote the corresponding induced single-vector instance . 
Let \(q \sim Q'\) and let \(\{D_1, \dots, D_n\}\) be the document sets in \(\mathcal{D}\). Assume the induced single-vector search problem is $c$-strongly stable. 

If the \emph{non-degeneracy condition}
$$c\cdot \max_k \min_{d \in D_k} \delta(q, d)
\ge \max_k \max_{d \in D_k} \delta(q, d)$$
holds, meaning that there exists at least one document set whose vectors are all sufficiently far from a given query vector, then the multi-vector search instance is stable.  

\end{lemma}

\begin{theorem}[Multi-Vector Stability with Chamfer Distance]
\label{thm:chamfer_stability}
Let \((\mathcal{Q}, \mathcal{D}, \delta, \mathrm{Chamfer})\) be a multi-vector search problem where $|Q| = k$ for some constant $k \ge 1$ for all $Q \in \mathcal{Q}$. Let $A_i = \min_{d \in D} \delta(q_i, d)$ be the nearest-neighbor distance for an individual query vector $q_i \in Q$. The problem is stable if: 
\begin{enumerate}
\item[(a)] The induced single-vector search instance is $c$-strongly stable. 

\item[(b)] The document sets satisfy non-degeneracy as defined in Lemma \ref{lem:limit_transfer}. 

\item[(c)] $\sum_{i < j} \Cov(A_i, A_j) \ge 0$.
\end{enumerate}

\end{theorem}

We provide the proofs of Lemma \ref{lem:limit_transfer} and Theorem \ref{thm:chamfer_stability} in Appendix \ref{sec:multi-vec-appendix}. In Appendix \ref{sec:empirical-chamfer-appendix}, we also verify that all three theorem assumptions comfortably hold true in practice. 

\subsection{Stability of Synthetic Embeddings}

To experimentally validate Theorem \ref{thm:chamfer_stability}, we generate a synthetic family of multi-vector search instances where we can take the dimensionality to be arbitrarily large. In Figure \ref{fig:multivec_stability} we report the results of our experiment with 1000 document sets and 100 query sets each with 4 vectors per set. Our goal in this construction is to generate a synthetic family of vectors sets that satisfy the strong stability, non-degeneracy condition, and non-negative covariances of the theorem while showcasing a case where average pooling exhibits instability while Chamfer aggregation does not. To that end, we generate query and document vectors sampled from a standard multivariate Gaussian distribution and add additional Gaussian noise to each vector. Crucially, for each vector $v$ we add to a document set, we also add the antipodal vector $-v$ to the set. This construction forces the average pooling aggregation with cosine distance to essentially ``cancel out" the signal and forcing all set distances to be essentially equivalent. However, Chamfer distance does not suffer from this issue since it will simply select the nearest neighbor document vector for each query vector. Thus, in Figure \ref{fig:multivec_stability}, we present a counterexample where, with all other theorem conditions being satisfied, Chamfer distance exhibits overwhelming stability while averaging results in instability.

\begin{figure}[H]
\centering
\includegraphics[width=1.05\linewidth]{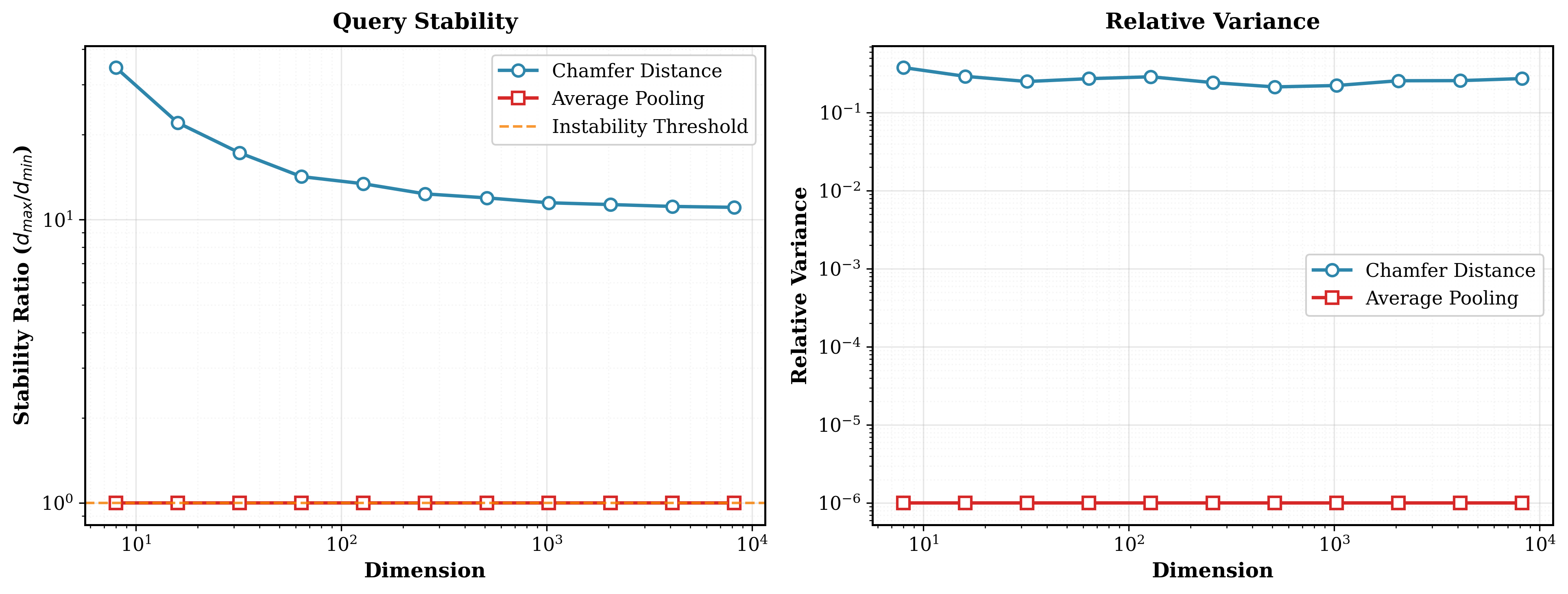}
\caption{We observe that multi-vector search with Chamfer distance maintains a non-decaying stability ratio and relative variance at high dimensions while average pooling, under the same vector set construction, exhibits instability.}
\label{fig:multivec_stability}
\end{figure}

\section{Stability of Filtered Vector Search}

We now move to the second major vector-based retrieval setting we analyze in this paper, namely filtered similarity search. As in the previous section, we will begin by defining the formal mathematical model for filtered vector search we consider and then introduce our main theorem. 







\subsection{Theoretical Model}

We begin with a formal definition of the filtered (or hybrid) near neighbor search problem which combines traditional vector similarity retrieval with attribute matching. 

\begin{definition}[Filtered Near-Neighbor Search] A \emph{filtered near-neighbor search problem} is a 6-tuple $(\mathcal{Q}, \mathcal{D}, \delta, A_Q, A_D, f)$ where:

\begin{itemize}
    \item $\mathcal{Q} = \{(q_1, A_{q_1}), (q_2, A_{q_2}), \dots \}$ is a set of query tuples where each $q_i \in \mathbb{R}^m$ and each $A_{q_i} \subseteq A_Q$ is a set of \emph{query attributes}.
    
    \item $\mathcal{D} = \{(d_1, A_{d_1}), (d_2, A_{d_2}), \dots \}$ is the set of database tuples where each $d_i \in \mathbb{R}^m$ and each $A_{d_i} \subseteq A_D$ is a set of \emph{document attributes}.

    \item $\delta: \mathbb{R}^m \times \mathbb{R}^m \rightarrow \mathbb{R}_{\ge 0}$ is a standard vector distance metric.

    \item $f\colon A_Q \times A_D \to \{0, 1\}$ is a \emph{filter function} that outputs 1 if the query and document attributes match and 0 otherwise. 
\end{itemize}

For a given query $(q, A_{q})$, we compute $$\argmin_{(d, A_d) \in \mathcal{D}} \delta(q, d) \quad \text{subject to } f(A_q, A_d) = 1$$

\end{definition}

\begin{remark}
This definition of filtered near neighbor search places no restriction on how the filter function $f$ operates. For instance, if the query and document attributes are discrete, then a natural definition of $f$ might be to output 1 if $A_q \subseteq A_d$ and 0 otherwise. Furthermore, our definition extends to other filtered search settings. In the \emph{approximate window search} problem \citep{engels2024}, each document is associated with a numeric label, and each query aims to find the nearest documents within arbitrary label ranges. In this case, we can model $A_q \subset \mathbb{R}$ as a continuous range of real numbers and $A_d \in \mathbb{R}$ as a scalar value. The filter function $f$ in this setting would output 1 if $A_d \in A_q$ and 0 otherwise.
\end{remark}

In this paper, we will focus our analysis on \emph{penalty-based} filtered near-neighbor search algorithms that augment the standard geometric vector distance metric with an additive penalty for points that do not satisfy the query filter. These functions are also referred to as \emph{fusion} distances in the literature \citep{wang2023efficient}. We will define the $\emph{penalized}$ distance function as $\delta^{\prime}(q, d) = \delta(q, d) + \alpha (1 - f(A_q, A_d))$ so that the distance metric incurs an additive scalar penalty of $\alpha$ if and only if $f(A_q, A_d) = 0$ meaning the query and document attributes do not match. 

We note that this formal model of fused distance functions captures nearly all filtered vector search algorithms in the literature. Some approaches, such as NHQ \citep{wang2023efficient} and CAPS \citep{gupta2023caps} set the penalty value via heuristics or parameter tuning. On the other hand, many approaches in the literature, such as Acorn \citep{patel2024acorn}, Filtered DiskANN \citep{gollapudi2023filtered}, and window search~\cite{engels2024} never consider points that violate the query filter predicate throughout the entire search process. These algorithms can also be modeled by our framework by setting $\alpha$ to be infinite. Thus, we emphasize that our formalization of filtered vector search closely aligns with nearly all algorithmic approaches in the literature.

\subsection{Theoretical Results}

We now present our main theorem of this section, which states that we can achieve stability in the filtered search setting by choosing a sufficiently large distance penalty value. In this theoretical analysis, we assume that the attributes are drawn from some distribution such that the filter mismatch probability is at most some constant $p_{\max}$ where $0<p_{\max}<1$.

\begin{theorem}[Stability of Filtered Search]
\label{thm:filtered-stability}
Consider the filtered near neighbor search setting with penalized distance functions. Suppose that the filter mismatch probability for any query-document pair is bounded by $0 < p_{\max} < 1$ for some constant $p_{\max}$. If the penalty value satisfies $\alpha > \frac{2 \Delta}{ 1 - p_{\max}}$ where $\Delta = \max_{q, d} \delta(q, d)$, then the filtered search problem is stable.
\end{theorem}

We defer the full proof of this theorem to Appendix \ref{sec:filtered-appendix}. We note that Theorem \ref{thm:filtered-stability} holds regardless of whether the underlying vector search problem (without constraints) is stable or unstable since our proof handles both cases directly. The latter case is particularly interesting as it implies that, under the assumptions of Theorem \ref{thm:filtered-stability}, the addition of filters can transform an \textbf{unstable} vector search problem into a \textbf{stable} one. 


\subsection{Experiments}

In this section, we empirically validate Theorem \ref{thm:filtered-stability} on a family synthetic vector search problems with 10,000 documents and 100 queries and varying dimensionality. We unit normalize all vectors so that the constant $\Delta \le 2$ and set the parameter $p_{\max}$ to be $0.5$ which satisfies the conditions of the theorem. Thus, as predicted by our theory, a sufficient condition for stability based on these parameters is setting the penalty value $\alpha$ to be greater than $8$. For this experiment, we start with a base vector search problem that is unstable and we consider three different choices for the penalty value: in the first case we take the penalty to be 8.1, namely slightly above the minimum threshold for stability; secondly, we set the penalty value to be smaller than the threshold suggested by Theorem \ref{thm:filtered-stability} and have it decay with the dimensionality; finally, we consider the case of a penalty value of 0. As suggested by the proof of Theorem \ref{thm:filtered-stability}, we want to engineer filter mismatches such that there is a negative correlation between the distance function and attribute matches. We simulate this phenomenon by setting filter matches and mismatches based on each document's average distance to all queries; documents that are close to queries on average become filter mismatches while documents further away are matches.

As predicted by Theorem \ref{thm:filtered-stability}, we see that the large penalty value above the theorem's threshold achieves stability. Moreover, the small, decreasing penalty value causes the problem to not satisfy the conditions of the theorem demonstrating that there exist non-zero penalty values that result in instability. Furthermore, when we set the penalty value to 0, the problem reduces to standard near-neighbor search where we observe instability by construction. Thus, this experiment validates that with a sufficiently high penalty value as prescribed by Theorem \ref{thm:filtered-stability}, an unstable near neighbor search problem becomes stable with the addition of filters. 

\begin{figure}[H]
\centering
\includegraphics[width=1.05\linewidth]{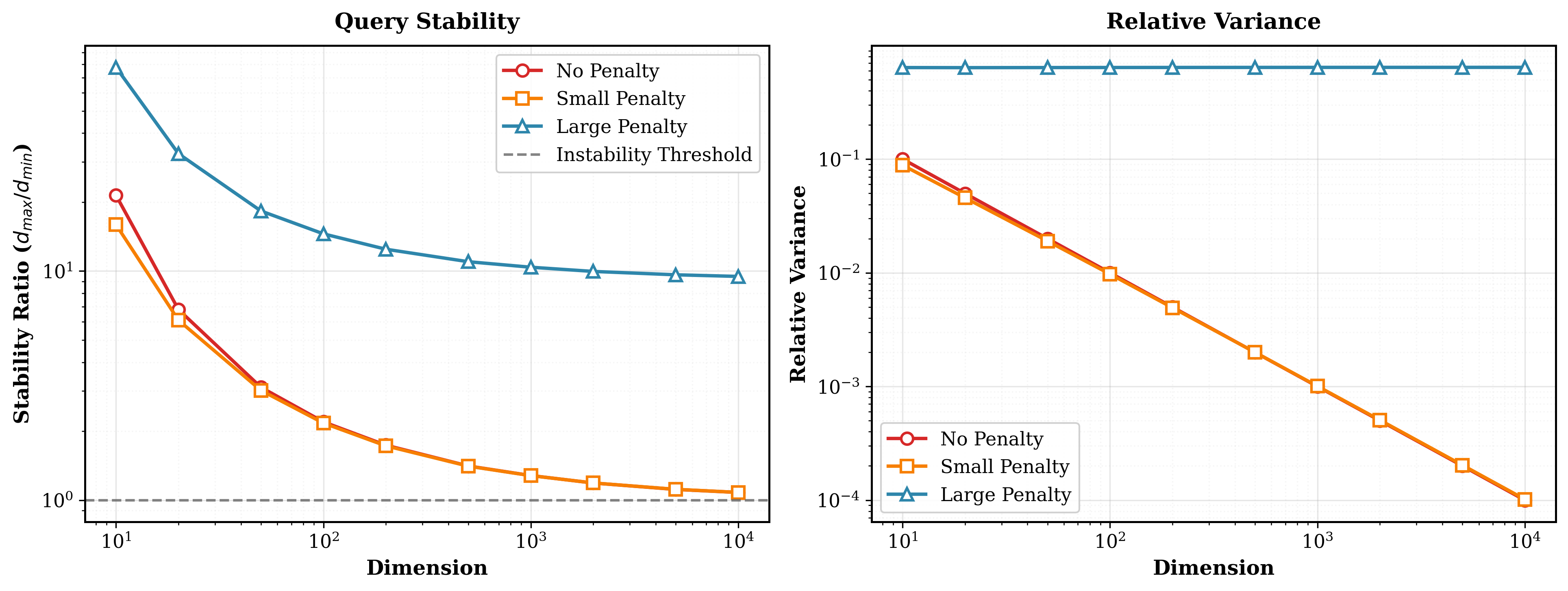}
\caption{We construct a filtered search instance that with a sufficiently large mismatch penalty exhibits stability but is unstable with a sufficiently small penalty and with no penalty at all.}
\label{fig:constraint_stability}
\end{figure}

\section{Stability of Sparse Vector Search}

We now turn our attention to the final vector search setting we study in this work: sparse embedding retrieval. 





\subsection{Theoretical Model}
In our setup, a sparse vector $d \in \mathbb{R}^m$ is characterized by its support $S_d \coloneqq \{i \in [m] \mid d_i \neq 0\}$. While in many sparse retrieval settings $|S_d| \ll m$ (often $|S_d| = O(\log m))$, for full generality we impose no growth rate on $|S_d|$ and only rely on the concentration and overlap of importance to establish stability. In other words, we do not restrict ourselves to any particular definition of sparsity. Our main theorem is more general in nature and may apply to certain dense vector settings as well. Nevertheless, the primary application of this theorem remains in the sparse setting. 

\begin{definition}[Concentration of Importance (CoI)]\label{coi-def}
Let $d\in \mathbb{R}_{\geq 0}^m$ and let $d_{(1)} \geq d_{(2)} \geq \dots \geq d_{(m)}$ be its coordinates sorted by value and denote by $||d||_p$ its $\ell_p$ norm. For a constant $\kappa \ll m$ where $\kappa $ independent of $m$ (i.e., not growing with the dimension), define 
\[
C_d(\kappa) \coloneqq \frac{\sum_{i=1}^{\kappa}|d_{(i)}|^p}{||d||_p^p} \in [0,1]
\]

Let $Q,D \subset \mathbb{R}_{\geq 0}^{m}$. We say that the pair $(Q,D)$ satisfies concentration of importance with parameters $(\kappa, \alpha, R, \rho)$ if
\begin{enumerate}
    \item[(a)] $C_q(\kappa) \geq \alpha$ for all $q \in Q$.
    \item[(b)] $\mathop{\Pr}_{d\sim D}\left[C_d(R\kappa) \ge \alpha \right] \ge \rho$.
    \item[(c)] $R\kappa \ll m$, where $R \geq 1$ is fixed and independent of $m$.
\end{enumerate}
\end{definition}

\begin{remark}
Restricting to nonnegative vectors entails no loss of generality. We provide a full proof of this claim in Lemma ~\ref{bi-lipschitz-lemma}. We state our results for non-negative vectors for notational simplicity. 
\end{remark}


\begin{definition}[Overlap of Importance] \label{overlap-def}
Fix $p\ge 1$ and $0 < \kappa \ll m$. For $d\in\R_{\ge 0}^m$, let $T_d$ be the indices of the
top-$\kappa$ coordinates of $d$ by magnitude. We say a $(Q,D,\delta)$ sparse vector
search problem exhibits \emph{overlap of importance} if there exist constants
$\gamma>0$ and $\pi \in(0,1)$ such that, for $q\in Q$ and $d\in D$ with
$\|q\|_p=\|d\|_p=1$,
\[
\Pr\!\Bigg[\, \sum_{i\in T_q\cap T_d} \min\{\,q_i^{\,p},\, d_i^{\,p}\,\}
\ \ge\ \gamma \,\Bigg] \ \ge\ \pi.
\]
\end{definition}

Intuitively, the overlap condition necessitates that with non-vanishing probability, the two vectors share a set of coordinates that cumulatively carry at least $\gamma$ in both. 

\subsection{Theoretical Results}

Given the sparse vector search model introduced in the previous section, we now state our main result via concentration and overlap of importance. 

\begin{restatable}{theorem}{MainSparseTheorem}\label{thm:stability-coi-overlap}

Let $(Q, D, ||\mathord{\cdot}||_p)$ be a near neighbor search problem where $Q, D \subset \mathbb{R}_{\geq  0}^{m}$ with $||q||_p = ||d||_p = 1$ for all $q \in Q, d \in D$  and $p \geq 1.$ Assume that

\begin{enumerate}

\item[(a)] $(Q,D)$ satisfy $\text{CoI}(\kappa_q, \alpha, R, \rho)$ as stated in Definition~\ref{coi-def}.
\item[(b)] $(Q,D)$ satisfy overlap of importance with parameters $(\gamma, \pi)$ as stated in Definition~\ref{overlap-def}.
\item[(c)] There exists a constant $\tau$ such that $\Pr[i \in T_x] \leq \frac{\tau}{m}$ for any $i \in [m]$ and $x \in \{d, q\}$.

\end{enumerate}

Under assumptions $(a)$--$(c)$ above, define $X \coloneqq (2 - 2\gamma)^{1/p}$ and $Y \coloneqq \frac{\rho 2^{1/p}}{1-\pi} \left(\alpha^{1/p} - (1-\alpha)^{1/p}\right)$ and assume that $Y > X$. Then, 

\[
\liminf_{m \to \infty}\RelVar_m \geq \frac{\pi}{4}(Y - X)^2 > 0.
\]

Hence, the sparse vector search problem is stable by the relative variance criterion. 

\end{restatable}

We defer the full proof of this theorem to Appendix \ref{sec:sparse-stability-appendix}. This result is, to our knowledge, the first to prove sufficient conditions for stability in the context of sparse vector search, which is a particularly interesting setting given the high dimensionalities of learned sparse embeddings, such as SPLADE, used in practice.

The above theorem demonstrates that sparse vector search is stable whenever $X = (2 - 2\gamma)^{1/p}$ is less than $Y = \frac{\rho 2^{1/p}}{1-\pi}\left(\alpha^{1/p} - (1-\alpha)^{1/p}\right)$. These two expressions represent different geometric configurations for the vectors $q$ and $d$. The quantity $X$ is an upper bound on the expected distance when $q$ and $d$ overlap substantially in their head coordinates. In this case, the vectors partially align and the distance shrinks. In contrast, $Y$ represents the expected distance in the opposite configuration where $T_q$ and $T_d$ are disjoint. Here the vectors place their largest coordinates in opposite directions, and hence the distance grows. 

Equality $X = Y$ corresponds to a single algebraic relationship between parameters $\alpha, \gamma$ and $\pi$ where the two geometric configurations happen to produce the exact same distance. Geometrically, this means that placing mass on overlapping head coordinates reduces the distance by precisely the same amount that shifting the mass into disjoint coordinates would increase it. This cancellation requires the parameters $(\alpha, \gamma, \pi)$ to satisfy a rigid relation and therefore occurs on a set of measure zero on the parameter space. Everywhere else for all practically relevant parameterizations, the two configurations remain geometrically distinct, which ensures that the distribution of the distances does not collapse in the limit, which we require for stability. Furthermore, in Appendix \ref{sec:empirical-sparse-appendix}, we empirically verify that the conditions of Theorem \ref{thm:stability-coi-overlap} hold true in practice on real SPLADE embeddings. 

\subsection{Stability of Synthetic Sparse Embeddings}

In this section, we evaluate Theorem~\ref{thm:stability-coi-overlap} using synthetic sparse embeddings constructed to reflect the structural patterns commonly observed in modern sparse neural embeddings. Two properties guide our design, namely \emph{(1) head concentration} where a small set of coordinates carries most of the $\ell_p$ mass and \emph{(2) semantic locality}, where high-importance dimensions cluster within coherent neighborhoods instead of being uniformly distributed across the embedding vector. To reproduce these behaviors, our data generator samples high-importance dimensions from a Zipf-shaped distribution and orients this selection towards a semantic center, which causes most of the mass to concentrate within a contiguous set of dimensions.


Figure~\ref{fig:sparse-stability} shows the results of the experiments when we set $\alpha = 0.83$, $\gamma = 0.2$, and $\pi = 0.5$ for the overlap probability. In particular, the CoI-only regime highlights that even strongly-concentrated embeddings exhibit distance collapse as the dimension grows when there is not enough overlap in their heads. These results confirm that neither concentration nor overlap of importance are sufficient conditions for stability on their own.

\begin{figure}[H]
\centering
\includegraphics[width=1.05\linewidth]{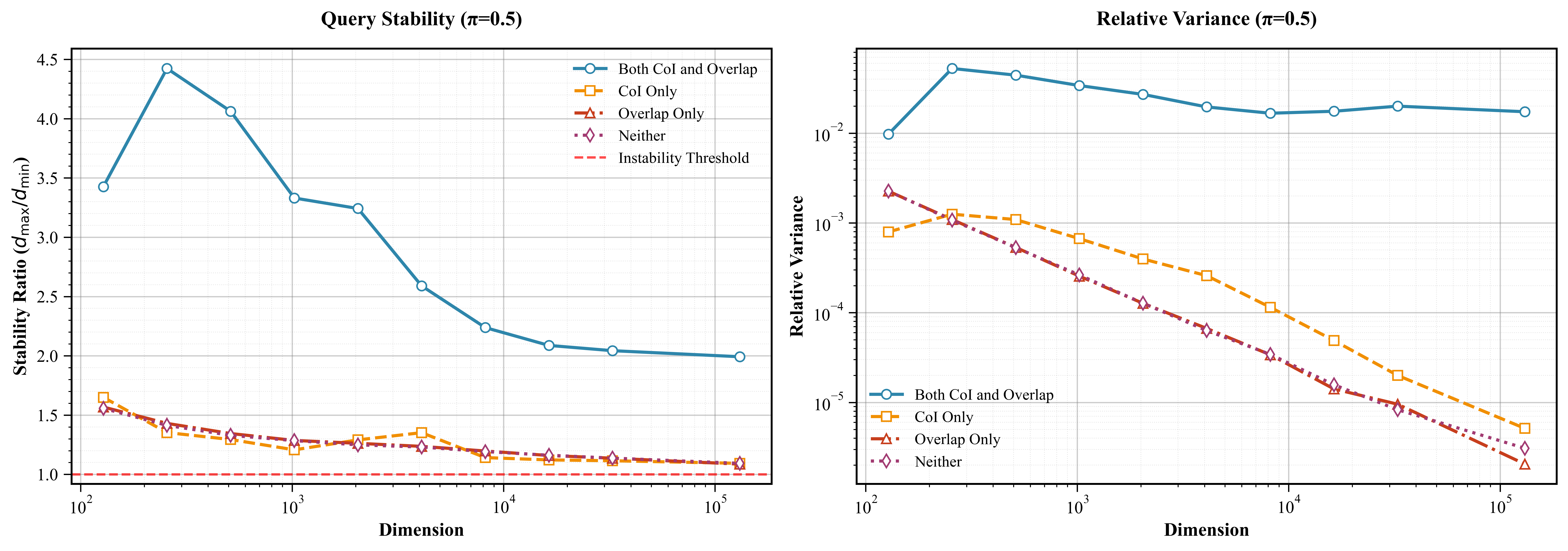}
\caption{We observe that search over sparse embeddings that violate either concentration or overlap of importance or both exhibits instability. However, when we fulfill both conditions we reach stability.}
\label{fig:sparse-stability}
\end{figure}

\section{Conclusion}

Modern vector search, driven by advances in neural representation learning, has achieved remarkable success in enabling information retrieval across a variety of modalities. However, the high dimensions of these embeddings suggest that vector retrieval over neural representations should suffer from the curse of dimensionality where the distances between all points tends to converge to a uniform value, thereby rendering approximate search algorithms as ineffective. This paper investigates why modern sublinear vector retrieval algorithms succeed in spite of the curse of dimensionality by building off foundational results on near-neighbor search \emph{stability}. We introduce formal models for three widely popular variants of vector retrieval, namely multi-vector search, filtered vector search, and sparse vector search, and identify sufficient conditions for which near-neighbor search remains stable in these settings. Crucially, we note that the mathematical techniques we present in this work are general in nature and can very likely be applied to other vector search settings as well. We believe this may constitute an exciting direction for future work. 

\bibliography{example_paper}
\bibliographystyle{icml2026}

\newpage
\appendix
\onecolumn
\newpage

\section{Extended Related Work: Intrinsic Dimensionality}
\label{sec:intrinsic-dim}
A closely related concept to stability is \emph{intrinsic dimensionality} (ID) which measures the dimensionality of the subspace that a collection of vectors lie within. Intuitively, if a database of high-dimensional vectors lies within a much lower-dimensional manifold, then a near-neighbor search over these points is insulated from the curse of dimensionality effect. Thus, much like stability, ID provides a foundational theoretical tool for explaining why modern vector retrieval is capable of circumventing the curse of dimensionality. Consequently, ID has been well-studied with multiple proposed formal definitions. Perhaps the most influential definition of ID is that of \cite{karger2002finding} via a geomeric property called the \emph{expansion rate}. Another popular definition heavily utilized in near-neighbor search algorithm design is the notion of \emph{doubling dimension} proposed by \cite{gupta2003bounded}. We refer interested readers to the recent book of \cite{bruch2024foundations} for the precise definitions. More recently, several works have measured the intrinsic dimensionality of modern neural embeddings and have reported that the ID is often much less than the ambient dimensionality. For instance, \cite{razzhigaev2024shape} study the intrinsic dimensionality of embeddings from transformer models \citep{vaswani2017attention} and find that the ID typically lies between 20 and 60 over the course of pretraining. In addition, \cite{gong2019intrinsic} compute the intrinsic dimensionality of popular image model embeddings, finding that the ID in practice typically ranges from 12 to 19, well below the ambient dimensions of 128 to 512. Moreover, \cite{ansuini2019intrinsic} go beyond embeddings and measure the ID of convolutional neural network (CNN) layers where they find the ID is an order of magnitude smaller than the ambient dimension. \\

Given that neural models consistently exhibit low intrinsic dimensionality, a natural question to ask is whether we still need to study why modern vector databases succeed in spite of the curse of dimensionality. On the contrary, we argue that intrinsic dimensionality alone is  insufficient to explain the phenomena we see in practice. First and foremost, as we empirically measure in Section \ref{sec:practical-implications}, a dimensionality as low as 32, which is within the empirical ID measurements of real embeddings, is still high enough for the curse of dimensionality to diminish the search quality of approximate near neighbor search algorithms. This observation motivates the need to study stability explicitly. Secondly, modern embeddings produced by large language models and used in practice continue to grow in dimension into the thousands and tens of thousands of components. It is not clear if the intrinsic dimensionality will remain low as these embedding models scale and yet we observe they remain amenable to sub-linear vector search techniques. Moreover, stability analysis provides actionable guidance that ID 
alone cannot. For instance, our multi-vector stability theorem explains 
why ColBERT's choice of Chamfer distance succeeds where average pooling 
might fail, a distinction invisible to ID measurements which only 
characterize the embedding space geometry, not the aggregation 
function. Similarly, our filtered search results suggest that 
practitioners can induce stability through sufficiently large 
penalties, even when the underlying embeddings are intrinsically 
high-dimensional.

\section{Practical Implications of Stability}
\label{sec:practical-implications}

In this section, we empirically demonstrate that stability is not merely a theoretical curiosity but also has substantial bearing on the quality of vector retrieval algorithms in practice. To our knowledge, this is the first empirical demonstration of the effects of instability on popular approximate ANN search algoritms. In Figure \ref{fig:algorithmic-stability}, we report the results of running two of the most popular ANN search techniques, namely the hierarchical navigable small worlds (HNSW) graph-based index \citep{malkov2018efficient} and the inverted file (IVF) partitioning-based index \citep{douze2024faiss}, across stable and unstable datasets. 

We fix a standard configuration of hyperparameters for each respective algorithm and compute the top-10 recall by taking the average of 1000 query vectors over a database of 1 million document vectors. To generate the unstable datasets at varying dimensionalities, we sample the value of each component independently from a univariate standard Gaussian distribution following the classic construction of \citep{beyer1999nearest}. To construct the stable datasets, we generate vectors such that they belong to one of five distinct clusters. For consistency, we fix the search index hyperparameters across all experiments; for HNSW, we set \texttt{ef-construction} and \texttt{ef-search} to be $200$ and $\texttt{M}=16$; for IVF, we set the number of clusters to $\texttt{nlist}=\sqrt{n}$, where $n$ is the number of database elements and the number of probes to $\texttt{nprobes}=\frac{\sqrt{n}}{4}$.

\begin{figure}[H]
\centering
\includegraphics[width=0.7\linewidth]{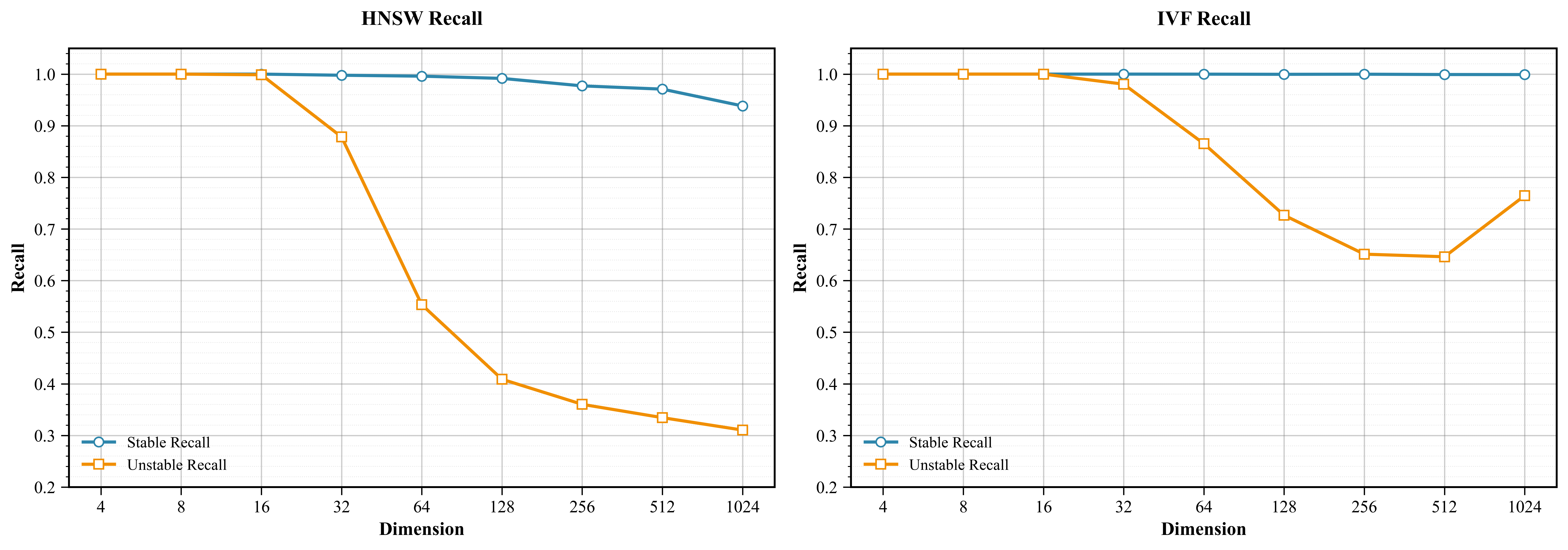}
\caption{Comparison of stable vs. unstable recall@10 across algorithms and dimensions. We observe the instability causes approximate search algorithms to lose recall in higher dimensions while stability preserves search quality as the dimensionality scales.}
\label{fig:algorithmic-stability}
\end{figure}

From Figure \ref{fig:algorithmic-stability}, we see that these approximate near-neighbor search algorithms severely suffer in recall over unstable databases since the lack of contrast between points renders sublinear algorithms that avoid a brute force scan ineffective. We also see that the impact of instability worsens in higher dimensions. These results underscore the importance of stability in practice.

\section{Stability of Multi-Vector Search}
\label{sec:multi-vec-appendix}
\subsection{Proofs}

\begin{restate}
\textbf{Theorem \ref{thm:multi-vec-equivalence}} 
For sufficiently large databases, a multi-vector search instance is stable if and only if $$\lim_{m \to \infty} \frac{\Var\left[\mathrm{Agg}(Q, D)\right]}{\E\left[\mathrm{Agg}(Q, D) \right]^2} \ne 0$$
\end{restate}

\begin{proof}
This result follows directly from the prior proofs of \citep{beyer1999nearest} and \citep{durrant2009nearest} since the original proofs require only that the distance metric is non-negative. We provide a full, self-contained argument here for completeness. 

As a matter of notation, let $\mathrm{Agg}(Q,D)$ be the random variable for the aggregated distance in an $m$-dimensional space. Let its expected value be $\mathbb{E}[\mathrm{Agg}(Q, D)]$ and its variance be $\text{Var}[\mathrm{Agg(Q, D)}]$. The database produces $n$ samples of this random variable. Let $\mathrm{DMIN} = \min_{D \in \mathcal{D}} \mathrm{Agg}(Q, D)$ and $\mathrm{DMAX} = \max_{D \in \mathcal{D}} \mathrm{Agg}(Q, D)$.

$\implies$ This direction follows the argument from \citep{beyer1999nearest}. We will show that if $\lim_{m\to\infty} \frac{\text{Var}[\mathrm{Agg}(Q, D)]}{\mathbb{E}[\mathrm{Agg}(Q, D)]^2} = 0$, then the multi-vector search problem is unstable. Define a normalized random variable $Y = \frac{\mathrm{Agg}(Q, D)}{\mathbb{E}[\mathrm{Agg}(Q, D)]}$. Its expectation is $\mathbb{E}[Y] = 1$ and its variance is $\text{Var}[Y] = \frac{\text{Var}[\mathrm{Agg}(Q, D)]}{\mathbb{E}[\mathrm{Agg}(Q, D)]^2}$ The condition is thus equivalent to $\lim_{m\to\infty} \text{Var}[Y] = 0$. Since $\lim_{m\to\infty} \mathbb{E}[Y] = 1$ and $\lim_{m\to\infty} \text{Var}[Y] = 0$, the random variable $Y$ converges in probability to 1. Let $\{Y_{1}, \dots, Y_{n}\}$ be the normalized distances for the $n$ items in the database. Since $\min()$ and $\max()$ are continuous functions, Slutsky's Theorem implies that $\min Y_{i} \to_p 1$ and $\max Y_{i} \to_p 1$. Since both the numerator and denominator converge in probability to 1, their ratio also converges in probability to 1. By definition, $\frac{\mathrm{DMAX}}{\mathrm{DMIN}} \to_p 1$ means that for any $\epsilon > 0$, we have $\lim_{m\to\infty} \Pr\left[\left|\frac{\mathrm{DMAX}}{\mathrm{DMIN}} - 1\right| \le \epsilon\right] = 1$. Since $\mathrm{DMAX} \ge \mathrm{DMIN}$, this is equivalent to $\lim_{m\to\infty} \Pr\left[\mathrm{DMAX} \le (1+\epsilon)\mathrm{DMIN}\right] = 1$. This is precisely the definition of an unstable search problem as stated in Definition \ref{def:multi-vec-stability}. \\

$\impliedby$ This converse result follows the proof presented by \citep{durrant2009nearest}. We proceed by proving the contrapositive: If the multi-vector search problem is unstable, then the relative variance must be zero. Assume the problem is unstable. Thus, for any $\epsilon > 0$, $\lim_{m\to\infty} \Pr\left[\mathrm{DMAX} \le (1+\epsilon)\mathrm{DMIN} \right] = 1$. This directly implies that the ratio of distances converges in probability to 1: $\frac{\mathrm{DMAX}}{\mathrm{DMIN}} \to_p 1$. For a sufficiently large database, we can assume that the theoretical expectation $\mathbb{E}[\mathrm{Agg}(Q, D)]$ lies between the sample minimum and maximum. This allows us to bound the normalized random variable $Y = \frac{\mathrm{Agg}(Q, D)}{\mathbb{E}[\mathrm{Agg}(Q, D)]}$ for any sample $D_i \in \mathcal{D}$ as follows:
$$ \frac{\mathrm{DMIN}}{\mathrm{DMAX}} \le \frac{\mathrm{Agg}(Q, D_i)}{\mathbb{E}[\mathrm{Agg}(Q, D_i)]} \le \frac{\mathrm{DMAX}}{\mathrm{DMIN}} $$

From the instability assumption, we know both the lower bound $\frac{\mathrm{DMIN}}{\mathrm{DMAX}}$ and the upper bound $\frac{\mathrm{DMAX}}{\mathrm{DMIN}}$ converge in probability to 1 By the Squeeze Theorem, the random variable $Y$ must also converge in probability to 1. A random variable that converges in probability to a constant must have a variance that converges to zero. Therefore:
    $$ \lim_{m\to\infty} \text{Var}[Y] = \lim_{m\to\infty} \frac{\text{Var}[\mathrm{Agg}(Q, D)]}{\mathbb{E}[\mathrm{Agg}(Q, D)]^2} = 0 $$

This completes the proof of the contrapositive.
\end{proof}

\begin{restate}

\textbf{Lemma \ref{lem:limit_transfer}}[Chamfer Stability with Singleton Query Sets]
Let \((\mathcal{Q}, \mathcal{D}, \delta, \mathrm{Chamfer})\) be a multi-vector search problem where $|Q_i| = 1$ for each query set $Q_i \in \mathcal{Q}$. Let \((Q', D', \delta)\) denote the corresponding induced single-vector instance . 
Let \(q \sim Q'\) and let \(\{D_1, \dots, D_n\}\) be the document sets in \(\mathcal{D}\).

Assume the induced single-vector search problem is $c$-strongly stable

If the \emph{non-degeneracy condition}
\[
c\cdot \max_k \min_{d \in D_k} \delta(q, d)
\;\ge\;
\max_k \max_{d \in D_k} \delta(q, d)
\]
holds, meaning that there exists at least one document set whose vectors are all sufficiently far from a given query vector, then the multi-vector search instance is stable.  

\end{restate}

\begin{proof}
Let
\[
\mathrm{DMIN} = \min_{d \in D^\prime} \delta(q, d),
\quad
\mathrm{DMAX} = \max_{d \in D^\prime} \delta(q, d).
\]
Also define
\[
\mathrm{DMIN}^{\prime} = \min_k \min_{d \in D_k} \delta(q, d) = \mathrm{DMIN},
\quad
\mathrm{DMAX}^{\prime} = \max_k \min_{d \in D_k} \delta(q, d).
\]
We thus have
\[
\frac{\mathrm{DMAX}^{\prime}}{\mathrm{DMIN}^{\prime}}
= \frac{\max_k \min_{d \in D_k} \delta(q, d)}{\mathrm{DMIN}}
\ge \frac{1}{c}\frac{\mathrm{DMAX}}{\mathrm{DMIN}} > \frac{1}{c}(c + \epsilon) > 1.
\]

for some constant $\epsilon > 0$ by $c$-strong stability. Thus, it follows that $\frac{\mathrm{DMAX}^{\prime}}{\mathrm{DMIN}^{\prime}} \not \to_p 1$, completing the proof.
\end{proof}

Before proceeding to our main theorem, we state one additional helper lemma. 

\begin{lemma}[Dan's Favorite Inequality \citep{Spielman2018}]
\label{lem:dan-favorite}

Let $a_1, \dots, a_n$ and $b_1, \dots, b_n$ be positive numbers. Then 

$$\min_i \frac{a_i}{b_i} \le \frac{\sum_i a_i}{\sum_i b_i} \le \max_i \frac{a_i}{b_i}$$

\end{lemma}

\begin{restate}
\textbf{Theorem \ref{thm:chamfer_stability}}
Let \((\mathcal{Q}, \mathcal{D}, \delta, \mathrm{Chamfer})\) be a multi-vector search problem where $|Q| = k$ for some constant $k \ge 1$ for all $Q \in \mathcal{Q}$. Let $A_i = \min_{d \in D} \delta(q_i, d)$ be the nearest-neighbor distance for an individual query vector $q_i \in Q$. The problem is stable if: 
\begin{enumerate}
\item[(a)] The induced single-vector search instance is $c$-strongly stable

\item[(b)] The document sets satisfy non-degeneracy as defined in Lemma \ref{lem:limit_transfer}

\item[(c)] $\sum_{i < j} \Cov(A_i, A_j) \ge 0$
\end{enumerate}
\end{restate}

\begin{proof}
We will show that 
\[
\lim_{m \to \infty} 
\frac{\Var[\mathrm{Chamfer}(Q,D)]}{\E[\mathrm{Chamfer}(Q,D)]^2} > 0
\]

which implies the desired result.

To begin, we observe that 

\[
\lim_{m \to \infty} 
\frac{\Var[\mathrm{Chamfer}(Q,D)]}{\E[\mathrm{Chamfer}(Q,D)]^2} = \lim_{m \to \infty} \frac{\Var\Big(\sum_{i=1}^k A_i\Big)}{\Big(\E\left[\sum_{i=1}^k A_i\right]\Big)^2} = \lim_{m \to \infty} \frac{\Var\Big(\sum_{i=1}^k A_i\Big)}{\Big(\sum_{i=1}^k \E[A_i] \Big)^2}
\]

where the second equality follows from the linearity of expectation. 

We note that we can lower bound the numerator of the above expression as follows:
\[
\Var\Big(\sum_{i=1}^k A_i\Big) = \sum_{i=1}^k \Var[A_i] + 2 \sum_{i < j} \Cov(A_i, A_j) \ge \sum_{i=1}^k \Var[A_i]
\]

where the inequality follows from the third assumption in the theorem statement. 

Furthermore, we can upper bound the term in the denominator via the Cauchy-Schwarz inequality

\[
\Big(\sum_{i=1}^k \E[A_i]\Big)^2 \le k \sum_{i=1}^k \E[A_i]^2.
\]

Putting these bounds together, we have 
\[
\lim_{m \to \infty} \frac{\Var[\sum_{i=1}^k A_i]}{\big(\sum_{i=1}^k \E[A_i]\big)^2} 
\ge \lim_{m \to \infty} \frac{\sum_{i=1}^k \Var[A_i]}{k \sum_{i=1}^k \E[A_i]^2} 
\ge \frac{1}{k} \lim_{m \to \infty} \min_i \frac{\Var[A_i]}{\E[A_i]^2} > 0
\]

where the last two inequalities follow from applying Lemma \ref{lem:dan-favorite} and Lemma \ref{lem:limit_transfer} respectively. This completes the proof. 

\end{proof}

\begin{remark}
Stability follows from applying Dan's Favorite Inequality (Lemma \ref{lem:dan-favorite}) to reduce the general multi-vector search problem to multi-vector search over singleton query sets, namely the setting in Lemma \ref{lem:limit_transfer}, at the cost of an $\frac{1}{k}$ overhead factor. The $1/k$ factor suggests that stability weakens as the query set size grows, which is also consistent with empirical observations.
\end{remark}

\subsection{Empirical Validation of Multi-Vector Theorem Assumptions}
\label{sec:empirical-chamfer-appendix}
Theorem \ref{thm:chamfer_stability} relies on three assumptions to prove the stability of multi-vector search under Chamfer distance: $c$-strong stability, non-degenerate document sets, and a non-negative bound on the sum of the covariances of the individual Chamfer distance terms. In this section, we experimentally demonstrate that all three of these assumptions are satisfied in practice, which underscores the practical utility of our result. Specifically, we use ColBERT embeddings \cite{khattab2020colbert} to construct query and document vector sets for five standard English-based information retrieval datasets from the BEIR benchmark \cite{thakur2beir}. For each dataset, we empirically estimate the strong stability constant $c$ as the smallest value of the ratio between the maximum distance and the minimum distance across all queries in the induced single-vector search instance. We then use this value of $c$ to check the proportion of document sets that satisfy the non-degeneracy condition. Finally, we then compute the sum of the covariances of the near-neighbor distances of individual query vectors in a given query set and verify that the sum is non-negative. This non-negative covariance assumption is natural in practice since the tokens in a query set tend to be topically related. We summarize our empirical measurements in Table \ref{tab:multivector-dataset-validation} where we see that all three theorem conditions are comfortably satisfied in practice.

\begin{table}[H]
\centering
\resizebox{0.5\linewidth}{!}{%
\begin{tabular}{@{}lccc@{}}
\toprule
\textbf{Dataset} & \makecell{\textbf{Stability} \\ \textbf{Constant} ($c$)} & \makecell{\textbf{Non-Degeneracy} \\ \textbf{Pass Rate (\%)}} & \makecell{\textbf{Cov. Bound} \\ \textbf{Pass Rate (\%)}} \\
\midrule
MS MARCO\footnotemark & 2.14 & 100.0 & 100.0 \\
Natural Questions & 2.36 & 100.0 & 100.0 \\
HotpotQA & 2.12 & 100.0 & 100.0 \\
TREC-COVID & 2.59 & 100.0 & 100.0 \\
NFCorpus & 1.91 & 100.0 & 100.0 \\
\midrule
\textbf{Requirement} & $> 1$ & $=100$ & $=100$ \\
\bottomrule \\
\end{tabular}
}%
\caption{Validation of Theorem~\ref{thm:chamfer_stability}'s assumptions across multiple datasets. All conditions are satisfied across all datasets, demonstrating the practical relevance of our theoretical results.}
\label{tab:multivector-dataset-validation}
\end{table}

\footnotetext{
For MS MARCO, we use the Passage Ranking dataset: \url{https://github.com/microsoft/MSMARCO-Passage-Ranking}.
}

\subsection{Stability of ColBERT Embeddings}

Following our empirical validation of Theorem \ref{thm:chamfer_stability}'s conditions in the previous section, we also compute the empirical stability ratio (maximum distance to minimum distance per query) and relative variance on a subset of the MSMarco dataset with ColBERT embeddings. We note that this experiment does not, strictly speaking, relate to the theoretical result since ColBERT embeddings maintain a fixed embedding size while stability is an asymptotic property. Nevertheless, we observe in Figure \ref{fig:colbert_stability} that the asymptotic behavior holds in practice at the ColBERT dimensionality of 768. Under Chamfer distance, the ColBERT vector sets are overwhelmingly more stable than with averaging, which aligns with the fact that we can rigorously prove stability under Chamfer distance but not with averaging. 

\begin{figure}[H]
\centering
\includegraphics[width=0.6\linewidth]{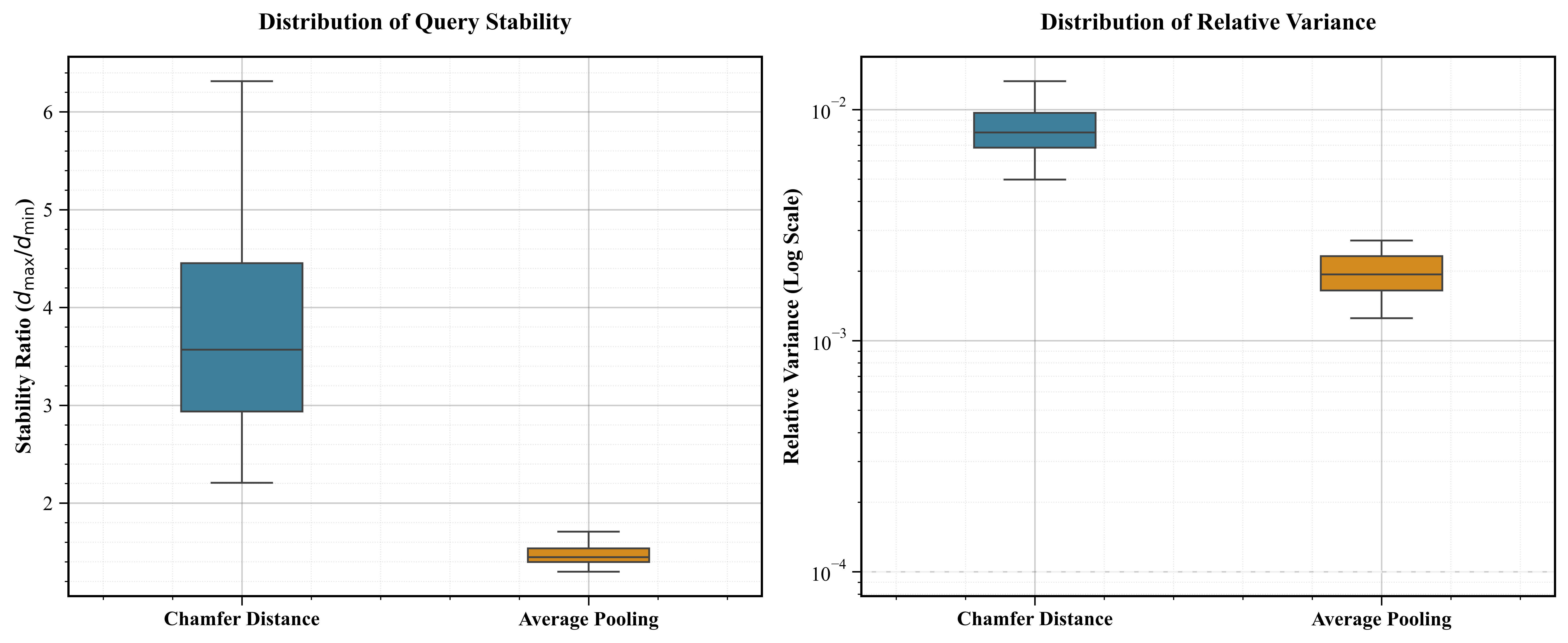}
\caption{Box plots of query stability and relative variance on the MSMarco dataset. We observe that Chamfer distance is more stable than average pooling, which supports our theoretical findings.}
\label{fig:colbert_stability}
\end{figure}

\section{Stability of Filtered Vector Search}
\label{sec:filtered-appendix}

\begin{restate}
\textbf{Theorem \ref{thm:filtered-stability}} 
Consider the filtered near neighbor search setting with penalized distance functions. Suppose that the filter mismatch probability for any query and document is bounded by $0 < p_{\max} < 1$ for some constant $p_{\max}$. If the penalty value satisfies $$\alpha > \frac{2 \Delta}{ 1 - p_{\max}}$$ where $\Delta = \max_{q \in Q, d \in D} \delta(q, d)$, then the filtered search problem is stable. 

\end{restate}

\begin{proof}
As a shorthand, let $\delta = \delta(q, d)$, let the indicator variable $1_c = \mathbf{1}{\{f(A_q, A_d) = 0\}}$ denote a filter mismatch, let and $p_c = \E[1_c]$. We analyze the limiting behavior of the Relative Variance of $\delta'$, denoted $\RelVar'$.

The expression for $\RelVar'$ can be derived as follows. First, the variance and expectation of $\delta'$ are:
$$ \E[\delta^\prime] = \E[\delta] + \alpha p_c $$
$$ \Var[\delta^\prime] = \Var[\delta] + \alpha^2 \Var[1_c] + 2 \alpha \Cov(\delta, 1_c) $$ 
Therefore,
\begin{align*}
\RelVar' &= \frac{\Var[\delta^\prime]}{\E[\delta^\prime]^2} \\
&= \frac{\Var[\delta] + \alpha^2 \Var[1_c] + 2 \alpha \Cov(\delta, 1_c)}{(\E[\delta] + \alpha p_c)^2}
\end{align*}

Dividing both the numerator and the denominator by $\E[\delta]^2$ gives us

\begin{align*}
    \RelVar' &= \frac{\frac{\Var[\delta]}{\E[\delta]^2} + \frac{\alpha^2 p_c(1-p_c)}{\E[\delta]^2} + 2 \alpha \frac{\Cov(\delta, 1_c)}{\E[\delta]^2}}{(1 + \frac{\alpha p_c}{\E[\delta]})^2} \\ \\
    &\ge \frac{\frac{\alpha^2 p_c(1-p_c)}{\Delta^2} + 2 \alpha \frac{\Cov(\delta, 1_c)}{\E[\delta]^2}}{(1 + \frac{\alpha p_c}{\E[\delta]})^2}
\end{align*}

where we use the facts that $\frac{\Var[d]}{\E[\delta]^2} \ge 0$ and $\Delta \ge \E[\delta]$

We next observe that
\begin{align*}
\frac{\Cov(\delta, 1_c)}{\E[\delta]} &= \frac{\E[\delta 1_c] - \E[\delta]\E[1_c]}{\E[\delta]}\\
&= \frac{\E[\delta 1_c]}{\E[\delta]} - \E[1_c]\\
&\geq 0 - p_c
\end{align*}
where the last inequality results from the fact that both $\delta \geq 0$ and $1_c \geq 0$. 

Applying this bound to our previous expression gives us 

\begin{align*}
    \RelVar^\prime \ge
    \frac{\frac{\alpha^2 p_c(1-p_c)}{\Delta^2} + 2 \alpha \frac{\Cov(d, 1_c)}{\Delta \E[d]}}{(1 + \frac{\alpha p_c}{\E[\delta]})^2}
    \ge \frac{\frac{\alpha^2 p_c(1-p_c)}{\Delta^2} - \frac{2 \alpha p_c}{\E[\delta]}}{(1 + \frac{\alpha p_c}{\E[\delta]})^2}
\end{align*}

Our goal is to show that this final expression is strictly positive, which is true if the following inequality holds

$$\frac{\alpha^2 p_c(1-p_c)}{\Delta^2} > \frac{2 \alpha p_c}{\E[\delta]} \ge \frac{2 \alpha p_c}{\Delta}$$

Thus, we can conclude that, if $\alpha > \frac{2 \Delta}{1 - p_{\max}}$, then $$ \lim_{m\to\infty} \RelVar' > 0$$ which completes the proof.
\end{proof}

\section{Stability of Sparse Vector Search}
\label{sec:sparse-stability-appendix}
\subsection{Proofs}
    
Recall that we are assuming $q, d \in \mathbb{R}_{\geq 0}^m$ and are normalized (i.e., $||q||_d = ||d||_d = 1$). We will start with the following lemma, which we use to assert that restricting our setting to $\mathbb{R}_{\geq 0}^m$ entails no loss of generality.

\begin{lemma}\label{bi-lipschitz-lemma}
Let $\phi : \mathbb{R}^m \to \mathbb{R}^{2m}$ be given by $\phi(x) = (x^+, x^-)$ with $x_i^+ = \max\{x_i, 0\}$ and $x_i^- = \max\{-x_i, 0\}.$ Then, for all $x,y \in \mathbb{R}^m$ and $p \geq 1,$
\[
2^{-(1-1/p)} ||x-y||_p \leq ||\phi(x) - \phi(y)||_p \leq 2^{1/p}||x-y||_p.
\]
\end{lemma}

\begin{proof}

For $x,y \in \mathbb{R}^m,$ define $\Delta_i \coloneqq |x_i - y_i|,\;\; A_i \coloneqq |x_i^+ - y_i^+|, \;\; B_i \coloneqq |x_i^- - y_i^-|$.

We will first show that for scalars $u,v \in \mathbb{R}$ and $p \geq 1$
\begin{align}
    2^{-(p-1)}|u-v|^p \leq |u^+ - v^+|^p + |u^- - v^-|^p \leq 2 |u - v|^p.\label{eq-scalar-inequality}
\end{align}

To see why, set $a = u^+ - v^+$ and $b = u^- - v^-$ and observe that $u - v = a - b$. Using Hölder’s inequality in $\mathbb{R}^2$ with conjugate exponents $p$ and $q = p/(p-1)$, we have that 
\[
|u-v| = |\langle(a,b), (1,-1)\rangle| \leq ||(a,b)||_p ||(1,-1)||_q.
\] 

Raising to the $p$-$\text{th}$ power, we get 
\[
|u-v|^p \leq (|a|^p + |b|^p)2^{p/q}.
\]
Since $p/q = p-1,$ we have that $2^{-(p-1)}|u-v|^p \leq |a|^p + |b|^p$. 

For the upper bound, it suffices to observe that the functions that map $t \to t^+$ and $t \to t^-$ for $t \in \mathbb{R}$ are both 1-Lipschitz. That is, $|u^+ - v^+| \leq |u - v|$ and $|u^- - v^-| \leq |u - v|$ from which we get $|u^+ - v^+|^p + |u^- - v^-|^p \leq 2 |u - v|^p$.

With Equation~\ref{eq-scalar-inequality} established, we can apply it coordinate-wise to $u = x_i, v=y_i$ to get 
\[
2^{-(p-1)}\sum_{i=1}^{m} \Delta_i^p \leq \sum_{i=1}^{m}(A_i^p + B_i^p) \leq 2 \sum_{i=1}^{m} \Delta_i^p.
\]

The result follows after taking the $p$-th root and observing that 
\[
||x-y||_p^p = \sum_{i}\Delta_i^p, \qquad ||\phi(x) - \phi(y)||_p^p = \sum_{i}(A_i^p + B_i^p).
\]
\end{proof}

We now proceed to the main theorem on sparse vector search.

\MainSparseTheorem*

\begin{proof}

Set $\delta(q,d) \coloneqq ||q - d||_p$ (for brevity, we will denote it as $\delta$) and define
\[
S \coloneqq \sum_{i\in T_q \cap T_d} \min\{q_i^p, d_i^p\}.
\]
Then, let $Z$ be the indicator random variable given by $Z \coloneqq \textbf{1}\{S \geq \gamma\} \in \{0,1\}$ and write the conditional means 
\[
\mu_1 \coloneqq \E[\delta \mid Z = 1], \qquad \mu_0 \coloneqq \E[\delta \mid Z = 0].
\]
 
Our first goal is to establish lower and upper bounds on $\mu_0$ and $\mu_1$, respectively.

First, note that for any $a \geq b \geq 0$ and $p \geq 1, (a - b)^p \leq a^p - b^p$. To see why this is true, note that 
\[
a^p - b^p = \int_{b}^a p x^{p-1}dx.
\]
For $p \geq 1,$ the mapping that sends $x$ to $x^{p-1}$ is non-decreasing on $[0,\infty).$ Since the integrand is monotonic, we get that 
\[
\int_{b}^{a}px^{p-1}dx \geq \int_{0}^{a-b}px^{p-1}dx = (a-b)^p.
\]

For each coordinate $i$, set $a = \max\{q_i, d_i\}$ and $b = \min\{q_i, d_i\}$. Now, using the above identity we note that 
\begin{align*}
    |q_i - d_i|^p &= (a - b)^p \\
    &\leq a^p - b^p \\
    &= (\max\{q_i, d_i\})^p - (\min\{q_i, d_i\})^p \\
    &= \left(\max\{q_i^p, d_i^p\} + \min\{q_i^p, d_i^p\}\right) - 2 \min \{q_i^p, d_i^p\} \\
    &= q_i^p + d_i^p - 2\min\{q_i^p, d_i^p\}.
\end{align*}
Summing over $i \in [m]$ yields
\[
\delta^p = \sum_{i=1}^{m}|q_i - d_i|^p \leq \sum_{i=1}^{m}(q_i^p+d_i^p) - 2 \sum_{i=1}^{m} \min\{q_i^p, d_i^p\}.
\]
By our normalization assumption and the definition of $S$, we get that $\delta^p \leq 2 - 2S$. 

When $Z= 1,$ we have that $S \geq \gamma.$ Thus, 
\[
\mu_1 = \E[\delta \mid Z = 1] \leq (2-2\gamma)^{1/p}.
\]

To get a lower bound on $\mu_0$, we will make use of assumption $(c)$. First, define $K= |T_q \cap T_d| = \sum_{i=1}^{m} \textbf{1}\{i \in T_q \cap T_d\}$ and observe that its expectation is given by 
\begin{align}
\E[K] &= \E\left[\sum_{i=1}^{m} \textbf{1}\{i\in T_q \cap T_d\}\right] \nonumber \\
&= \sum_{i=1}^{m} \E\left[\textbf{1}\{i\in T_q \cap T_d\}\right] \nonumber \\
&=\sum_{i=1}^{m} \Pr[i \in T_q]\Pr[i\in T_d] \nonumber \\
&\leq \sum_{i=1}^{m} \left(\frac{\tau}{m}\right)^2 \nonumber \\
&= \frac{\tau^2}{m}. \label{eq:upper_bound_on_Ek}
\end{align}

Now for any index set $T \subseteq [m],$ define 
\[
x_i^{(T)} =
\begin{cases}
    x_i, & i \in T,\\
    0 & \text{otherwise}
\end{cases}
\qquad
x \in \{q, d\}.
\]
Then, for any such $T$, 
\begin{align}\label{eq:rev-triangle-inequality}
\left(\sum_{i \in T} |q_i - d_i|^p\right)^{1/p} = ||q^{(T)} - d^{(T)}||_p \geq \left|||q^{(T)}||_p - ||d^{(T)}||_p\right|.
\end{align}
by the reverse triangle inequality.

Since $T_q$ indexes the top $\kappa$ coordinates of $q$, concentration of importance for queries implies 
\[
||q^{T_q}||_p^p = \sum_{i\in T_q}q_i^p \geq \alpha.
\]
Let $C \coloneqq \{C_d(R\kappa) \geq \alpha\}$ denote the document concentration event, which holds with probability at least $\rho > 0$ by assumption. On $C,$
\[
||d^{(T_d)}||_p^p = \sum_{i\in T_d} d_i^p \geq \alpha.
\]

Now consider the event $\{K = 0\}$ under which $T_q \cap T_d = \emptyset$. In this case, the $\alpha$ head mass of $q$ lies entirely outside $T_d,$ and the $\alpha$ head mass of $d$ lies entirely outside $T_q.$ Thus, 

\[
\sum_{i\in T_q}q_i^p \geq \alpha, \qquad \sum_{i\in T_q}d_i^p \leq 1 - \alpha.
\]
Similarly,
\[
\sum_{i\in T_d}d_i^p \geq \alpha, \qquad \sum_{i\in T_d}q_i^p \leq 1 - \alpha.
\]

Applying Equation~\ref{eq:rev-triangle-inequality} to both $T_q$ and $T_d$ yields
\[
\sum_{i\in T_q}|q_i - d_i|^p \geq \left(\alpha^{1/p}-(1-\alpha)^{1/p}\right)^p, \qquad \sum_{i\in T_d}|q_i - d_i|^p \geq \left(\alpha^{1/p}-(1-\alpha)^{1/p}\right)^p.
\]
on the event $\{K=0\} \cap C.$

Therefore, on this event, we obtain 
\[
\delta^p = \sum_{i=1}^{m}|q_i - d_i|^p \geq \sum_{i\in T_q}|q_i - d_i|^p + \sum_{i \in T_d}|q_i - d_i|^p \geq 2\left(\alpha^{1/p}-(1-\alpha)^{1/p}\right)^p.
\]
Hence, taking the $p$-th root and expectation, we obtain 
\begin{align}\label{eq:expectation-lower-bound}
\E[\delta \mid Z = 0, K = 0, C] \geq 2^{1/p}\left(\alpha^{1/p}-(1-\alpha)^{1/p}\right).
\end{align}

Using the law of total expectation, we can decompose $\mu_0$ to get 
\[
\mu_0 = \E[\delta \mid Z= 0] \geq \E[\delta \mid Z=0, \{K=0\} \cap C]\Pr[\{K=0\} \cap C \mid Z=0].
\]

To lower bound $\Pr[\{K=0\} \cap C \mid Z=0], $ observe that $K=0$ implies $Z=0$ (since in this case the overlap condition does not hold). Thus, $\{\{K = 0\} \cap C\} \subseteq \{Z=0\}$ and the conditional probability simplifies to 

\[
\Pr[\{K=0 \cap C\} \mid Z =0] = \frac{\Pr[\{K=0\} \cap C]}{\Pr[Z=0]}.
\]

Using a union bound on the complementary event, note that
\[
\Pr[(\{K=0\}\cap C)^c] = \Pr[\{K > 0\} \cup C^c] \leq \Pr[K > 0] + \Pr[C^c].
\]
Using assumption $(c)$, the fact that $\Pr[K>0] \leq \tau^2/m$ (from Equation~\ref{eq:upper_bound_on_Ek} and Markov's inequality), and by probabilistic concentration of importance on the documents, $\Pr[C] \geq \rho$ so $\Pr[C^c] \leq 1 - \rho.$ Thus, 
\[
\Pr[\{K=0\} \cap C] \geq 1 - \left(\frac{\tau^2}{m} + 1 - \rho \right) = \rho - \frac{\tau^2}{m}.
\]

The overlap assumption (i.e., $\Pr[Z =0] \leq 1 - \pi$) combined with the above result yields
\[
\Pr[\{K=0\} \cap C \mid Z = 0] \leq \frac{\rho - \tau^2/m}{ 1 - \pi}.
\]

Substituting Equation~\ref{eq:expectation-lower-bound} and the above result into the expression for $\mu_0$ we get that
\begin{align}
\mu_0 \geq \frac{\rho - \tau^2/m}{1-\pi} \left[2^{1/p} \left( \alpha^{1/p} - (1-\alpha)^{1/p}\right) \right] \eqqcolon Y_m. \label{eq:lower_bound_on_mu_0}
\end{align}

Now, with a lower bound on $\mu_0$ and an upper bound on $\mu_1,$ we can conclude the proof. Recall that by the law of total variance 
\begin{align*}
\Var[\delta] &= \Var[\E[\delta \mid Z]] + \E[\Var[\delta \mid Z]] \\
&\geq \Var[\E[\delta \mid Z]] \\
&= \Pr[Z=1]\Pr[Z=0](\mu_1 - \mu_0)^2.
\end{align*}

Observe that $Z=1$ implies that $K > 0$, from which we get $\Pr[Z = 1] \leq \Pr[K > 0] \leq \frac{\tau^2}{m}.$ Thus, $\Pr[Z=0] \geq 1 - \frac{\tau^2}{m}$. Since $0 \leq \delta \leq 2$, $\E[\delta] \leq 2.$ Therefore, combining the above probability bounds yields 
\[
\RelVar_m = \frac{\Var[\delta]}{\E[\delta]^2} \geq \frac{\pi}{4}\left( 1 - \frac{\tau^2}{m}\right)(\mu_1 - \mu_0)^2.
\]

Now, note that $\mu_0 - \mu_1 \geq Y_m - X$. Since $Y_m \to Y$ and by assumption $Y > X,$ there exists an integer $m_0$ such that $Y_m > X$ for all $m \geq m_0$.

Whenever $m  \geq m_0, \mu_0 - \mu_1 \geq Y_m - X > 0$, and therefore $
|\mu_1 - \mu_0| = \mu_0 - \mu_1
$. Next, fix any $\epsilon > 0$. Since $Y_m \to Y,$ there exists $m_1$ such that $|Y_m - Y| \leq \epsilon$. In particular, for all $m \geq m_1,$ $Y_m \geq Y - \epsilon$, which is equivalent to $Y_m - X \geq (Y-X) - \epsilon.$ Thus, for all $m \geq \max\{m_0, m_1\}$, the following holds
\[
|\mu_1 - \mu_0| = \mu_0 - \mu_1 \geq Y_m - X \geq (Y - X) - \epsilon
\]
Since $\epsilon > 0$ is arbitrary, this shows that 
\[
\liminf_{m \to \infty} |\mu_1 - \mu_0| \geq Y - X > 0.
\]

Therefore, 
\begin{align*}
\liminf_{m \to \infty}\RelVar_m & \geq \frac{\pi}{4} \liminf_{m \to \infty}\left( (1 - \tau^2/m)(\mu_1 - \mu_0)^2\right) \\
&= \frac{\pi}{4}\left(\lim_{m\to\infty} (1 - \tau^2/m)\right)\left(\liminf_{m \to \infty} (\mu_1 -\mu_0)^2\right) \\
&\geq \frac{\pi}{4} (Y-X)^2 \\
&> 0
\end{align*}

By Theorem \ref{durrantconverse}, strictly positive relative variance implies stability of the sparse vector search problem, which concludes the proof.

\end{proof}

\begin{remark}
We note out that we state the theorem using $\liminf \RelVar_m$ instead of $\lim \RelVar_m$. This is convenient because bounding $\liminf$ away from zero is sufficient to prove stability without needing to prove that the relative variance terms converge, which may or may not happen. 
\end{remark}

\subsection{Empirical Validation of Sparse Theorem Assumptions}
\label{sec:empirical-sparse-appendix}

We empirically validate the assumptions of Theorem \ref{thm:stability-coi-overlap} on $30,522$-dimensional SPLADE-V3~\cite{lassance2024splade} embeddings across several benchmark datasets. We verify concentration of importance, overlap of importance, and finally check the gap condition $Y > X$ required by the theorem. 

\paragraph{Concentration of importance.} We fix a query head size $\kappa$ and examine document head sizes of the form $R\kappa$ where $R$ is a small integer at most 8. For each candidate $R,$ we estimate $\rho = \Pr[C_d(R\kappa) \geq \alpha]$ over a random sample of documents and select the smallest value such that $\rho$ exceeds a target level, which we set to $0.9$. The resulting $(\kappa, R, \rho)$ values are reported in Table ~\ref{tab:sparse-theorem-constants}.

\paragraph{Overlap of importance.} For random query-document pairs we compute the overlap statistic 
\[
S = \sum_{i \in T_q \cap T_d} \min \{q_i^p, d_i^p\}
\]
Empirically we measure $\pi_{\max} = \Pr[S > 0]$, the fraction of pairs with strictly positive overlap. We set $\gamma$ to be the $(1 - \pi_{\max})$ empirical quantile of $S,$ i.e., the smallest positive overlap observed. The resulting tail mass is $\hat{\pi} = \Pr[S \geq \gamma]$. This guarantees that $\gamma > 0$ while using the largest feasible overlap tail mass supported by the data.  

Since the theorem requires $Y > X$, we compute these expressions for each dataset and report the value of $Y - X$ in Table ~\ref{tab:sparse-theorem-constants}. Similarly, to check that stability holds, we compute the stability ratios over sampled queries and report the values in the table. 

Among all datasets, the MS MARCO passage ranking dataset demonstrates the weakest empirical overlap of importance. The overlap tail mass $\hat{\pi}$ is considerably smaller and the smallest positive overlaps are very close to zero. To satisfy the condition that $Y > X,$ we increase the concentration parameters $\alpha$ and $\rho$ to meet the theorem gap requirement when overlap is extremely rare.

Overall, the experiments confirm strong concentration for both queries and documents, non-vanishing overlap with $\gamma > 0$ and $\hat{\pi} > 0,$ and strictly positive theoretical gap $Y - X$, collectively confirming that the assumptions of the theorem holds on real-world sparse embeddings. 

\begin{table}[t]
\centering
\footnotesize
\begin{tabular*}{\linewidth}{@{\extracolsep{\fill}} lcccccccc}
\toprule
\textbf{Corpus}
& $\alpha$
& $\kappa$
& $R$ 
& $\rho$
& $\gamma$
& $\hat{\pi}$ 
& $d_{\max}/d_{\min}$ 
& $Y - X$ \\
\midrule
HotpotQA & 0.85 & 24 & 3 & 0.9167 & 0.00157 & 0.5194 & 3.6966 & 0.0291 \\
TREC-COVID & 0.85 & 24 & 6 & 0.9443 & 0.00236 & 0.8213 & 4.7285 & 2.5830 \\
Natural-Questions & 0.85 & 18 & 6 & 0.9537 & 0.00137& 0.5335 & 2.9501 & 0.1324 \\
MS MARCO & 0.95 & 18 & 8 & 0.9691 & 0.00049 & 0.3267 & 3.2981 & 0.1159 \\
NFCorpus & 0.85 & 18 & 8 & 0.9615 & 0.0012 & 0.5412 & 2.3999 & 0.1713 \\
\bottomrule \\
\end{tabular*}
\caption{
Empirical theorem constants for SPLADE embeddings under $\ell_2$ distance. Here $\rho = \Pr[C_d(R\kappa) \geq \alpha]$ denotes the document CoI pass rate. The overlap tail threshold $\gamma$ is chosen as the largest feasible quantile level for which the empirical tail mass remains $\hat{\pi} = \Pr[S \geq \gamma]$ remains positive. $d_{\max}/d_{\min}$ is the mean stability ratio over the sampled queries. Across all datasets we observe strong concentration and controlled overlap and stability ratios well above the instability threshold, confirming the applicability of the stated theorem. Since we operate over $\ell_2$ distance, we fix the parameter $p=2$ in this analysis. 
\label{tab:sparse-theorem-constants}
}

\end{table}


\end{document}